\documentclass{IEEEtran4PSCC}
\ifCLASSINFOpdf
   \usepackage[pdftex]{graphicx}
\else
   \usepackage[dvips]{graphicx}
\fi
\usepackage[cmex10]{amsmath}

\usepackage{multicol}
\usepackage{amssymb, amsthm}
\usepackage[capitalize,nameinlink,noabbrev]{cleveref}
\usepackage{bm}
\usepackage{tikz}
\usetikzlibrary{calc,positioning}

\newcommand{\proj}{\operatorname{proj}}

\newtheorem{remark}{Remark}
\newtheorem{theorem}{Theorem}
\newtheorem{lemma}{Lemma}

\newtheorem{definition}{Definition}

\definecolor{paperBlue}{HTML}{0072B2}      
\definecolor{paperVermilion}{HTML}{D55E00} 
\definecolor{paperGreen}{HTML}{009E73}
\definecolor{paperSky}{HTML}{56B4E9}
\definecolor{paperPurple}{HTML}{CC79A7}

\makeatletter
\let\old@ps@headings\ps@headings
\let\old@ps@IEEEtitlepagestyle\ps@IEEEtitlepagestyle
\def\psccfooter#1{%
    \def\ps@headings{%
        \old@ps@headings%
        \def\@oddfoot{\strut\hfill#1\hfill\strut}%
        \def\@evenfoot{\strut\hfill#1\hfill\strut}%
    }%
    \def\ps@IEEEtitlepagestyle{%
        \old@ps@IEEEtitlepagestyle%
        \def\@oddfoot{\strut\hfill#1\hfill\strut}%
        \def\@evenfoot{\strut\hfill#1\hfill\strut}%
    }%
    \ps@headings%
}
\makeatother

\usepackage{titlesec}

\usepackage{xcolor}

\begin{document}

\title{Exact Recourse Functions for Aggregations of EVs Operating in Imbalance Markets}

\author{\IEEEauthorblockN{Karan Mukhi\IEEEauthorrefmark{1},
Licio Romao\IEEEauthorrefmark{2},
Alessandro Abate\IEEEauthorrefmark{1}}
\IEEEauthorblockA{\IEEEauthorrefmark{1} Department of Computer Science, University of Oxford, United Kingdom}
\IEEEauthorblockA{\IEEEauthorrefmark{2} Department of Wind and Energy Systems, Technical University of Denmark, Denmark}
}

\maketitle

\begin{abstract}
We study optimal charging of large electric vehicle populations that are exposed to a single real-time \emph{imbalance price}. The problem is naturally cast as a multistage stochastic linear programme (MSLP), which can be solved by algorithms such as Stochastic Dual Dynamic Programming. However, these methods scale poorly with the number of devices and stages.
This paper presents a novel approach to overcome this curse of dimensionality. Building prior work that characterises the aggregate flexibility sets of populations of EVs as a \emph{permutahdron}, we reformulate the original problem in terms of aggregated quantities.
The geometric structure of permutahedra lets us (i) construct an optimal disaggregation policy, (ii) derive an exact, lower-dimensional MSLP, and (iii) characterise the expected recourse function as piecewise affine with a \emph{finite, explicit} partition. In particular, we provide closed-form expressions for the slopes and intercepts of each affine region via truncated expectations of future prices, yielding an exact form for the recourse function and first-stage policy. 
Comprehensive numerical studies validate our claims and demonstrate the practical utility of this work.
\end{abstract}


\section{Introduction}
Distributed energy resources (DERs), such as electric vehicles, batteries, and energy storage systems, can be used as a way to increase the resilience of power grids with high penetration of renewable generation \cite{Lund2015ReviewElectricity}.
DERs can serve as a source of flexibility, to mitigate the intermittency of renewable resources, offering active and reactive power to help balance supply and demand in the power grid. Without these such flexible resources, the frequency and voltage of the entire system would violate admissible regions, thus undermining the well-functioning of the entire grid.
Populations of DERs can be used as cheap source of flexibility on the grid. Hence, to cope with the intermittency of energy production, Transmission System Operators (TSOs) are going to great lengths to encourage the participation of DERs in reserve markets.

To capitalize on the flexibility available in electric vehicles, aggregators play a crucial role for two reasons: (1) from the consumer perspective, they provide a means for EV owners to participate in the balancing of the power grid; (2) from a TSO perspective, due to the strong law of large numbers \cite{Vandael2015ReinforcementMarket}, the variance of the aggregated population decreases as the number of EVs under control increases. By bidding into reserve markets, aggregators expose themselves to imbalance prices and other penalties imposed by TSOs if they fail to deliver the committed capacity. Reaching a trade-off between such penalties and the revenue acquired by engaging in reserve markets constitutes a challenge aggregators must embrace. In this context, the decision problem faced by the aggregator is that of optimizing the consumption of each device in its population while maximizing profit and managing the risk.

Adjusting the charging rate of each EV based on the predicted imbalance price and revenue allows aggregators to maximize their revenue while controlling risks. A common formulation of this decision problem is as a multistage stochastic linear programme (MSLP), and a common algorithm to obtain an approximation of such a decision problem is stochastic dual dynamic programming (SDDP). SDDP produces a lower bound on the optimal recourse function and chooses actions that maximise such a minorant function \cite{Pereira1991Multi-stagePlanning}. However, a key limitation of SDDP is its poor scalability with respect to the dimensionality of the state space \cite{Fullner2025StochasticReview}. Applying SDDP to large-dimensional problems leads to insurmountable computational burden, thus limiting its applicability in real-world problems.

In this paper, we build on a corpus of work that focuses on characterising the aggregate flexibility set of a population of DERs \cite{Barot2017APolytopes}. Using these aggregation techniques we reformulate the MSLP in terms of a set aggregated states and decision variables. In the case of a population of EVs the aggregated flexibility set is a type of polytope known as a \emph{permutahedron} \cite{Mukhi2023AnVehicles}. By exploiting the geometry of this class of polytope \cite{Postnikov2009PermutohedraBeyond}, we show how a transition function and stagewise feasible sets can be defined in this aggregated domain. 
 Then, using the results in \cite{Forcier2024ExactProblems}, we compute the regions of the aggregated decision space where the optimal recourse function is affine. This allows us to find an exact representation of the recourse function, by deriving closed-form solutions of the coefficients that define each affine region. 
 Overall, our main contributions can be summarised as follows:

\begin{itemize}
    \item[1.] We reformulate the stochastic programming problem faced by EV aggregators in terms of aggregated charging and state of charge of the population, showing how this lower-dimensional representation exhibits important geometric properties;
    \item[2.] We leverage existing results in combinatorics to characterise the regions of the aggregated state space where the optimal recourse function is affine, and we demonstrate how the parameters of these regions can be obtained from the original problem formulation;
    \item[3.] We validate our findings through numerical experiments, showing that our exact characterisations of the recourse function is sound and scalable.
\end{itemize}

This paper is organised as follows. Section II contains preliminary results on MSLPs, SDDP, and permutahedra. Section III presents the problem formulation, and Section IV the proposed reformulation in terms of the aggregated quantities. The characterization of the affine regions of the optimal recourse function is presented in Section V, and the numerical results are described in Section VI.




\section{Preliminaries}
This section gathers the preliminaries used throughout the paper. We first formalise \textit{multistage stochastic linear programs} (MSLPs). We then introduce their dynamic-programming decomposition and briefly recall \textit{stochastic dual dynamic programming} (SDDP) as a method of solving this class of problems. 
Finally, we define a family of polytopes known as \textit{permutahedra} and state some properties of these polytopes that will become relevant later.

\subsection{Multistage Stochastic Linear Problems}
Let $(\Omega,\mathcal{A},\mathbb{P})$ be a probability space, and let $\{c_t\}_{t=1}^T$ be a sequence of random variables. We denote by $\mathcal{F}_t$ the $\sigma$-algebra generated by the sequence $\sigma(c_1,\ldots,c_t)$. Given an initial condition, $x_0$, an MSLP can be formalised as follows:
\begin{equation*}
\begin{aligned}
v^*(x_0) = \min_{(u_t)_{t=1}^T}\quad
& c_1^\top u_1 \;+\; \mathbb{E}\!\left[\sum_{t=2}^T c_t u_t \right]\\[2pt]
\text{s.t.}\quad
& u_t \in \mathcal{U}_t(x_{t-1})\\[2pt]
& x_t = T_t(x_{t-1},u_t) \qquad  \\[2pt]
& u_t \preceq \mathcal{F}_t \qquad && t \in \mathcal{T}.
\end{aligned}
\end{equation*}
Here, $u_t \in \mathbb{R}^{m_t}$ and $x_t \in \mathbb{R}^{n_t}$, denote the $t^{th}$ decision and state variables. The set $\mathcal{U}_t(x_{t-1})$ is the admissible action set given the incoming state, and the system evolves via the transition function $T_t(x_{t-1},u_t)$. The non-anticipativity condition $u_t \preceq \mathcal{F}_t$ enforces that $u_t$ is $\mathcal{F}_t$-measurable. The MSLP formulated above is a \textit{Hazard-Decision} problem, where the stage-t decision is made after observing the noise $c_t$ \cite{Dowson2020TheProblems}. To formulate it as a \textit{Decision-Hazard} problem, where the stage-t decision is made before observing $c_t$, one must modify the non-anticipativity constraint such that $u_t$ is measurable with respect to $\mathcal{F}_{t-1}$.

\begin{definition}
The stage-$t$ recourse function is defined recursively as:
\begin{equation*}
\begin{aligned}
    V_t(x_{t-1}, c_t) = \min_{u_t \in \mathbb{R}^{n_t}}\quad &c_t^\top u_t + \mathcal{V}_{t+1}(x_t) \\[4pt]
\text{s.t.}\quad
& u_t \in \mathcal{U}_t(x_{t-1}), \\[2pt]
& x_t = T(x_{t-1}, u_t), \\[2pt]
\end{aligned}
\end{equation*}
where the expected recourse function is given by
\begin{equation*} 
    \mathcal{V}_{t}(x_{t-1}) := \mathbb{E}_{c_t}\!\left[ V_t(x_{t-1}, c_t)  \right].
\end{equation*}
\end{definition}
\noindent
Using this, the MSLP may be written as
\begin{equation*}
\min_{u_1\in\mathbb{R}^{n_1}}\big\{\,c_1^\top u_1+\mathcal V_{2}(x_1) : u_1\in\mathcal U_1(x_0),\ x_1=T(x_0,u_1)\,\big\},
\end{equation*}
with the terminal condition $\mathcal{V}_{T+1}(x_T) = 0.$

\begin{figure}[t]
\centering
     \begin{tikzpicture}[scale=0.9]

  \def\xmin{0}\def\xmax{6}
  \def\ymin{-0}\def\ymax{6}

  \def\xA{1}
  \def\xB{3}
  \def\xC{5}

  \def\yA{2/3}
  \def\yB{2.5}
  \def\yC{45/8}

  \def\c{0}
  \def\mA{-2}\def\bA{5.0+\c}
  \def\mB{-.5}\def\bB{3.5+\c}
  \def\mC{.5}\def\bC{.5+\c}
  \def\mD{2}\def\bD{-7+\c}

  \def\d{0}
  \def\nA{-2.2}\def\cA{4.5+\d}
  \def\nB{-.7}\def\cB{3.5+\d}
  \def\nC{.5}\def\cC{.5+\d}
  \def\nD{1.3}\def\cD{-4+\d}

  \draw[->] (\xmin-0.5,\ymin) -- (\xmax+0.5,\ymin) node[midway, below=6pt] {$x_t$};
  \draw[->] (\xmin-0.5,\ymin) -- (\xmin-0.5,\ymax+0.2)  node[midway, rotate=90, above] {$\mathcal{V}_{t+1}(x_t)$};

  \coordinate (E0) at (\xmin,{\mA*\xmin+\bA});   
  \coordinate (E1) at (\xA,  {\mA*\xA+\bA});     
  \coordinate (E2) at (\xB,  {\mB*\xB+\bB});     
  \coordinate (E3) at (\xC,  {\mC*\xC+\bC});     
  \coordinate (E4) at (\xmax,{\mD*\xmax+\bD});   

    ===== TRUE convex PWA function as a polyline =====
  \coordinate (D0) at (\xmin, {\nA*\xmin+\cA});   
  \coordinate (D1) at (\yA,   {\nA*\yA+\cA});     
  \coordinate (D2) at (\yB,   {\nB*\yB+\cB});     
  \coordinate (D3) at (\yC,   {\nC*\yC+\cC});     
  \coordinate (D4) at (\xmax, {\nD*\xmax+\cD});   

   \coordinate (D1A) at (\xmin, \nB*\xmin + \cB);
    \coordinate (D1B) at ({(\ymin - \cA)/\nA}, \ymin);

   \coordinate (D2A) at (\xmin, \nC*\xmin + \cC);
    \coordinate (D2B) at ({(\ymin - \cB)/\nB}, \ymin);

    \coordinate (D3B) at ({(\ymin - \cD)/\nD}, \ymin);

   \coordinate (D4A) at (\xmax, \nC*\xmax + \cC);


    \draw[paperVermilion, dashed, opacity=0.5, line width=0.3pt] (D1A) -- (D1);
    \draw[paperVermilion, dashed, opacity=0.5, line width=0.3pt] (D1B) -- (D1);
    \draw[paperVermilion, dashed, opacity=0.5, line width=0.3pt] (D2A) -- (D2);
    \draw[paperVermilion, dashed, opacity=0.5, line width=0.3pt] (D2B) -- (D2);
    \draw[paperVermilion, dashed, opacity=0.5, line width=0.3pt] (D3B) -- (D3);
    \draw[paperVermilion, dashed, opacity=0.5, line width=0.3pt] (D4A) -- (D3);


  \draw[paperVermilion, line width=.7pt] (D0) -- (D1) -- (D2) -- (D3) -- (D4);

  \draw[paperBlue, line width=.7pt] (E0) -- (E1) -- (E2);

 \draw[paperBlue, line width=.7pt] (E2) -- (E3);
    
  \draw[paperBlue, line width=.7pt] (E3) -- (E4);

  \foreach \P in {E0,E1,E2,E3,E4}{
    \fill[paperBlue] (\P) circle (1.5pt);
}
\foreach \x in {\xmin, \xA, \xB, \xC, \xmax}{
    \draw[paperBlue, line width=0.7pt, cap=round]
        (\x, \ymin+0) -- (\x, \ymin+0.1);
}

  \foreach \P in {D0,D1,D2,D3,D4}{
    \fill[paperVermilion] (\P) circle (1.5pt);
}

\foreach \x in {\xmin, \yA, \yB, \yC, \xmax}{
    \draw[paperVermilion, line width=0.7pt, cap=round]
        (\x, \ymin) -- (\x, \ymin - 0.1);
}
\draw[densely dotted, paperBlue!60] (E0) --  (\xmin, \ymin);
\draw[densely dotted, paperBlue!60] (E1) --  (\xA, \ymin);
\draw[densely dotted, paperBlue!60] (E2) --  (\xB, \ymin);
\draw[densely dotted, paperBlue!60] (E3) --  (\xC, \ymin);
\draw[densely dotted, paperBlue!60] (E4) --  (\xmax, \ymin);

\draw[densely dotted, red!60] (D0) --  (\xmin, \ymin);
\draw[densely dotted, red!60] (D1) --  (\yA, \ymin);
\draw[densely dotted, red!60] (D2) --  (\yB, \ymin);
\draw[densely dotted, red!60] (D3) --  (\yC, \ymin);
\draw[densely dotted, red!60] (D4) --  (\xmax, \ymin);

\def\offset{0.02}
\foreach \strt/\ennd in {
    \xmin/\xA,
    \xA/\xB,
    \xB/\xC,
    \xC/\xmax,
}{
  \draw[paperBlue, line width=1pt]
    (\strt,\ymin+\offset) -- (\ennd,\ymin+\offset);
}

\foreach \strt/\ennd in {
    \xmin/\yA,
    \yA/\yB,
    \yB/\yC,
    \yC/\xmax,
}{
  \draw[paperVermilion, line width=1pt]
    (\strt,\ymin-\offset) -- (\ennd,\ymin-\offset);
}

\end{tikzpicture}
\caption{Exact expected recourse \(\mathcal V_{t+1}(x_t)\) (solid blue) is piecewise-affine and convex. The orange curve is a polyhedral \emph{minorant} built from a finite set of cuts as in SDDP (dashed lines show individual cuts); vertical guides mark the breakpoints of the affine regions. This paper identifies these regions and their coefficients in closed form, enabling direct construction of the exact value function.}

\label{fig:pwa-convex}
\end{figure}
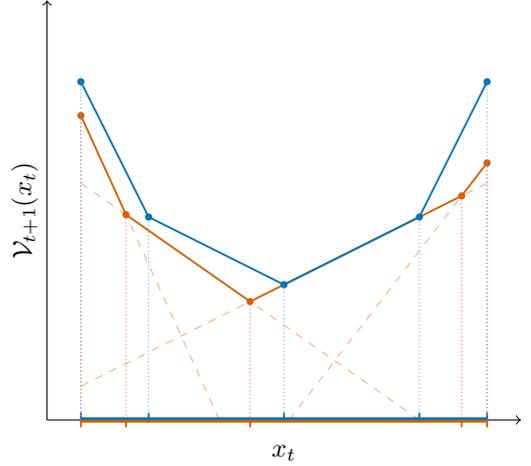

\subsection{Stochastic Dual Dynamic Programming}
Here we provide a very brief overview of SDDP, emphasizing aspects most relevant to this work. The reader is directed to \cite{Fullner2025StochasticReview} for a full review of this subject. 
SDDP provides a method of approximating the expected recourse function, and uses these approximations to solve the MSLP.
As formulated above, the expected recourse function $\mathcal{V}_{t+1}(x_t)$ is piecewise affine \cite[Theorem 2.8]{Fullner2025StochasticReview}, and so 
$\mathcal{V}_{t+1}(x_t)$ can be written as the maximum of a set of piecewise affine functions:
\begin{equation*}
     \mathcal{V}_{t+1}(x_t) = \max \left\{ \alpha_{t,k} + \sum_{i=1}^{n_t}\beta^i_{t, k} \;x^i_t \; \middle\vert \; k \in \mathcal{K}\right\},
\end{equation*}
SDDP constructs a sequence of polyhedral under-estimators of the expected recourse functions. In each iteration, a forward pass samples scenarios and generates trial states under the current approximation of the recourse function. A backward pass then solves the stage subproblems at those states and uses their dual variables to add supporting hyperplanes (“cuts”) to $\mathcal{V}_t$. Collectively, these cuts yield a lower bound that is valid for all states and nondecreasing across iterations; as more cuts are added, the approximation tightens and converges to the true expected recourse function \cite{Pereira1991Multi-stagePlanning}.

\begin{remark}\label{rem:affine_region}
    There exists a finite partitioning of the domain of $\mathcal{V}_{t+1}(x_t)$, denoted as $\mathcal{C} =\{\mathcal{R}_1, ..., \mathcal{R}_K|  \mathcal{R}_k \subseteq \mathbb{R}^{n_t} \;\; \forall k\}$, such that $\mathcal{V}_{t+1}(x_t)$ affine on each region of this partitioning. That is  $\forall k$ there exists $\alpha_{t, k}$ and $\beta_{t,k} \in \mathbb{R}^{n_t}$ such that $\forall x_t \in \mathcal{R}_k$
    \begin{equation*}
        \mathcal{V}_{t+1}(x_t) = \alpha_{t,k} + \sum_{i=1}^{n_t}\beta^i_{t, k} \;x^i_t ,
    \end{equation*}
    as shown in \cref{fig:pwa-convex}.
    The purpose of this paper is to find such a partitioning for the problem of a population of EVs exposed to stochastic imbalance prices. 
\end{remark}

\subsection{Permutahedra} 
Let $v \in \mathbb{R}^n$, such that $v^{i} \leq v^{i+1} \; \forall i = 1,...,n$. The \textit{permutahedra} generated by $v$, denoted $\Pi(v)$ is the convex hull of all coordinate permutations of $v$, that is \cite{Dahl2010Majorization01-matrices}:
\begin{equation*}   
    \Pi(v):=\mathrm{ConvHull}\{v^{\pi}: \pi \in S_n\},
\end{equation*}
where $v_{\pi}$ denotes the vector obtained by permuting the entries of $v$ according to $\pi$, and $S_n$ is the symmetric group on $n$ \cite{Rotman1995SymmetricG-Sets}. We briefly summarise some relevant properties of permutahedra below. For a comprehensive treatment, the reader is referred to \cite{Postnikov2009PermutohedraBeyond}.
Permutahedra are closed under Minkowski summation: the sum of two permutahedra is simply the permutahedron of the sum of their generators:
\begin{equation}\label{eq:perm_mink}
\Pi(v_a) \oplus \Pi(v_b) \;=\; \Pi(v_a + v_b).
\end{equation}
Given generators $v_a$ and $v_b$, $ \Pi(v_a)$ is a subset of $ \Pi(v_b)$, if and only if $v_b$ \textit{majorises} $v_a$, that is:

\begin{equation}\label{eq:majorization}
\Pi(v_a)\subseteq \Pi(v_b)
\;\Longleftrightarrow\; v_a \preceq v_b.
\end{equation}
Because of the symmetry arising from their construction, LPs over permutahedra admit an efficient solution:
\begin{equation}\label{eq:perm_lp}
    \sum_{i=1}^n w^{\pi^*(i)} v^i = \min\{w^\top u \mid u \in \Pi(v)\}
\end{equation}
where $\pi^*$ denotes the permutation that arranges the elements of $w$ in descending order: $  w^{\pi^*(i)} \geq  w^{\pi^*(i+1)}$ for all $i$. The optimal solution is $u^* = v^{\pi^*}$, given by permuting elements of $v$ using this same optimal permutation. Essentially solving an LP over $\Pi(v)$ reduces to sorting the objective coefficients.
Finally, the projection of a permutahedron onto any of its coordinates is given by the interval:
\begin{equation}\label{eq:perm_proj}
    \proj_{\{i\}}\left(\Pi(v)\right) = [v^1, v^{n}] \quad \forall i \in [n].
\end{equation}


\section{Problem Formulation}
In this section, we present our EV charging model and formulate the optimisation problem for a population of EVs exposed to stochastic imbalance prices.
We consider an aggregator with direct control over the power consumption of a finite population of devices indexed by
$l \in \mathcal{N}:=\{1,\ldots,N\}$.
The planning horizon is finite and discretised into $T$ stages of duration $\delta$.
We denote the set of stages by $\mathcal{T}:=\{1,\ldots,T\}$ and use $t\in\mathcal{T}$ to index a stage. Finally, for notational convenience we define $n_t = T - t + 1$.

\subsection{EV Charging Model}
Let $u_{t,l}$ denote the power consumption of the $l^{th}$ EV in stage $t$. Each EV will have a lower and upper bound on its power consumption, denoted $\underline{u}_l$ and $\overline{u}_l$ such that for all $i$ and $t$: $\underline{u}_l \leq u_{t,l} \leq \overline{u}_{i}$.
For the purposes of this work we assume $0 \leq \underline{u}_l$, i.e. EVs can only charge with no bidirectional charging capabilities. Although restrictive, this assumption reflects current practice, many vehicles are not V2G-enabled and many owners prefer not to export, so the formulation remains practically relevant.
Note, EVs with V2G enabled are treated in \cite{Mukhi2025ExactPolymatroids}, using geometric structures related to the one considered in this paper.

We let $x_{t,l}$ denote the remaining required energy demand for the EV at the end of stage $t$, where $x_{0,l}$ represents the device's initial energy requirement. For simplicity, and without loss of generality, in the following we assume $\underline{u}_l = 0$. In each stage, the remaining required energy demand decrements with the power consumption and the dynamics of the system are: 
\begin{equation}\label{eq:single_ev_transition}
    x_{t,l} = x_{t-1, l} - u_{t,l} \;\; \forall t,
\end{equation}
where we assume the duration of each stage $\delta = 1$. Without loss of generality, constant charging inefficiencies can be absorbed into the definition of $x_{t,l}$. The boundary condition $x_{t,l} = 0$ ensures the device is fully charged at the end of the horizon. 
Using this charging model, we may consider the set of feasible charging profiles for the device.
\begin{definition}[Individual recourse set]
For stage $t\in\mathcal T$, and device $l\in\mathcal N$, the \emph{individual recourse set}, denoted $\mathcal{P}_{t,l}(x_{t-1, l})$, is the set of all feasible \emph{future} charging profiles:
\begin{equation*}
     \mathcal{P}_{t,l}(x_{t-1, l}) :=  \left\{ u_l \in \mathbb{R}^{n_t} \middle\vert \;\;
    \begin{array}{@{}cl}
        0 \leq u_{s,l} \leq \overline{u}_l \;\; \forall s \in [n_t]\\
        \sum_{s=1}^{n_t} u_{s,l} = x_{t-1, l}
    \end{array} 
    \right\}
\end{equation*}
\end{definition}
By construction, the individual recourse sets are defined by a set of bounded linear constraints; and so the $\mathcal{P}_{t,l}(x_{t-1,l})$ are polytopes. In particular, $\mathcal{P}_{1,l}(x_{0,l})$ denotes the flexibility set for the device over the entire horizon, this is exactly the characterisation defined in \cite{Mukhi2023AnVehicles} and \cite{Panda2024EfficientVehicles}. It is convenient to define $q_{t,l} = \left\lfloor  x_{t-1, l}/\overline{u}_l \right\rfloor$, which represents the maximum number of remaining stages during which device can charge at its rated power.

\begin{lemma}\label{thm:perm_is_flex}\cite[Theorem 1]{Mukhi2023AnVehicles}
    The individual recourse set $\mathcal{P}_{t,l}(x_{t-1, l})$, is the permutahedra $\Pi(v_{t-1,l})$, where $v_{t-1,l} \in \mathbb{R}^{n_t}$ and 
    \begin{equation}\label{eq:vertex_major}
    v_{t-1,l}  := ( \underbrace{0, ...,\; 0}_{n_t  - q_{t,l}  - 1}, \;x_{t-1, l} - q_{t,l} \; \overline{u}_l , \;\underbrace{\overline{u}_l ,...,\overline{u}_l }_{q_{t,l}} ).
\end{equation}
\end{lemma}
Finally, the feasible instantaneous power for the device at stage $t$ is the projection of its recourse polytope onto the $t$-th coordinate: 
\begin{equation}\label{eq:single_ev_admissible}
    \mathcal{U}_t(x_{t-1,i}) = \proj_{\{t\}}\big(\mathcal{P}_{t,l}(x_{t-1, l})\big) = [v_{t-1,l}^1, v_{t-1,l}^{n_t}],
\end{equation}
where we have used the identity in \eqref{eq:perm_proj}. This characterises the feasible values of $u_{t,l}$ given $x_{t-1, l}$, and enforces the relatively complete recourse constraints required for SDDP \cite{Fullner2025StochasticReview}.

\subsection{Imbalance Market MSLP}

We formalise the online energy consumption problem faced by an aggregator participating in the real-time imbalance market under a single system price. Specifically, we consider a setting in which the entire population of EVs managed by the aggregator is uniformly exposed to a single imbalance price at each market interval, the aggregator must minimise the cost of charging the population whilst ensuring all EVs are fully charged at the end of the horizon.
 The price process $\{c_t\}$ is stage-wise independent, where stage-dependent marginals are derived from imbalance price forecasts such as \cite{Ganesh2024ForecastingNetworks}. Stage-wise independence of the imbalance price is a strong assumption, we adopt it for simplicity. However, it may be relaxed without altering the core contributions of this work.

At stage $t$, the aggregator observes the state of each device, $\bm{x}_{t-1} = (x_{t-1,l})_{l\in\mathcal N}$, and the current system price, $c_t$, and chooses a set of dispatch actions, $(u_{t,l})_{l\in\mathcal N}$, for each EV based on this. This can be formalised using the following MSLP
\begin{equation*}
\begin{aligned}
v^*(\bm{x}_0) = \min_{[u_{t,l}]^{l \in \mathcal{N}}_{t \in \mathcal{T}}}\quad
& c_1 \sum_{l \in \mathcal{N}}u_{1,l} \;+\; \mathbb{E}\!\left[\sum_{t=2}^T c_t \sum_{l \in \mathcal{N}}u_{t,l}  \right]\\[4pt]
\text{s.t.}\quad
& u_{t,l} \in [v_{t-1,l}^1, v_{t-1,l}^{n_t}], \\[2pt]
& x_{t,l} = x_{t-1, l} - u_{t,l} \\
& u_{t,l} \preceq \mathcal{F}_t   \quad\qquad\qquad  t=1,\ldots,T, \;\; \forall i.
\end{aligned}
\end{equation*}
Here, we use the admissible action set from \eqref{eq:single_ev_admissible} and the transition function from \eqref{eq:single_ev_transition}.
One could, in principle, apply SDDP to this formulation. However, there are $N$ control variables per stage, and SDDP scales poorly with the number of decision variables; as $N$ grows the subproblem and cut-generation burden explodes, rendering the application to large populations computationally intractable. Because each EV is exposed to the same market price, the objective depends only on aggregate power. This motivates an aggregated reformulation that collapses the $N$-dimensional control to a single aggregate action.

\section{Aggregated Reformulation}
In this section, we recast the MSLP in aggregate coordinates.  We first introduce a set of aggregate state variables and an aggregate charging decision that exactly characterise the set of admissible aggregate charging profiles.
We provide a policy that disaggregates an aggregate power consumption into 
into a set of admissible dispatch actions for each EV in the population. This policy induces a transition function for the aggregate state, the explicit form of which is derived below. Together with the associated admissible action set and transition, we rewrite the MSLP over the aggregated state and decision, yielding an equivalent but lower-dimensional MSLP.

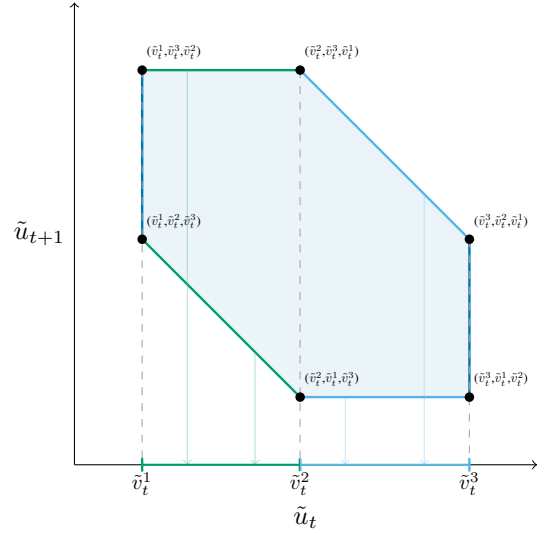
\begin{figure}[t]
    \centering
     \def\va{1.1}   
\def\vb{2.5}   
\def\vc{4}   

\def\na{$\tilde{v}^1_t$}
\def\nb{$\tilde{v}^2_t$}
\def\nc{$\tilde{v}^3_t$}

\begin{tikzpicture}[scale=1.5]

  \pgfmathsetmacro{\xmin}{min(\va,min(\vb,\vc))}
  \pgfmathsetmacro{\xmax}{max(\va,max(\vb,\vc))}
  \pgfmathsetmacro{\ymin}{\xmin}
  \pgfmathsetmacro{\ymax}{\xmax}
  \pgfmathsetmacro{\marg}{0.6}

  \draw[->] (0.5,0.5) -- (4.6,0.5) node[midway,  yshift=-20] {$\tilde{u}_t$};
  \draw[->] (0.5,0.5) -- (0.5,4.6) node[midway, left]  {$\tilde{u}_{t+1}$};


  \foreach \t in {\va,\vb,\vc} {
  
    \draw[gray!70]  ( \t,0.45) -- ( \t,0.55) node[gray, below=4pt] {};
  }

  \coordinate (Pabc) at (\va,\vb); 
  \coordinate (Pbac) at (\vb,\va); 
  \coordinate (Pbca) at (\vb,\vc); 
  \coordinate (Pcba) at (\vc,\vb); 
  \coordinate (Pcab) at (\vc,\va); 
  \coordinate (Pacb) at (\va,\vc); 

  \fill[paperBlue!8]  (Pbac) -- (Pcab) -- (Pcba) -- (Pbca) -- (Pacb) -- (Pabc) -- cycle;
  \draw[paperSky, line width=0.9pt](Pbac) -- (Pcab);
  \draw[paperBlue, line width=0.9pt](Pcab) -- (Pcba);
  \draw[paperSky, line width=0.9pt](Pcba) -- (Pbca);
  \draw[paperGreen, line width=0.9pt](Pbca) -- (Pacb);
  \draw[paperBlue, line width=0.9pt](Pacb) -- (Pabc);
  \draw[paperGreen, line width=0.9pt](Pabc) -- (Pbac);

\def\offset{0.4} 

  \draw[->, opacity=0.45, paperGreen!70] (\va + \offset, \vc) -- (\va + \offset ,0.5) node[black, below, scale=0.8]{};
  \draw[->, opacity=0.45, paperGreen!70] (\vb - \offset, \va + \offset) -- (\vb - \offset ,0.5) node[black, below, scale=0.8]{};

  \draw[->, opacity=0.45, paperSky!70] (\vb + \offset, \va) -- (\vb + \offset ,0.5) node[black, below, scale=0.8]{};
  \draw[->, opacity=0.45, paperSky!70] (\vc - \offset, \vb +\offset) -- (\vc - \offset ,0.5) node[black, below, scale=0.8]{};

  \draw[dashed,gray!70] (\va , \vc) -- (\va,0.5) node[black, below, scale=0.8] {\na};
  \draw[dashed,gray!70] (\vb, \vc) -- (\vb,0.5) node[black, below, scale=0.8] {\nb};
  \draw[dashed,gray!70] (\vc, \vb) -- (\vc,0.5) node[black, below, scale=0.8] {\nc};

  \draw[thick, paperGreen] (\va, 0.5) -- (\vb,0.5) node[midway, below, scale=0.8] {};
  \draw[thick, paperSky] (\vb, 0.5) -- (\vc,0.5) node[below, scale=0.8] {};
  

  \foreach \p/\name/\xshift/\yshift in {
      Pabc/{(\na,\nb,\nc)}/12/7,
      Pbac/{(\nb,\na,\nc)}/12/7,
      Pbca/{(\nb,\nc,\na)}/12/7,
      Pcba/{(\nc,\nb,\na)}/12/7,
      Pcab/{(\nc,\na,\nb)}/12/7,
      Pacb/{(\na,\nc,\nb)}/12/7}
  {
    \fill (\p) circle (1.2pt);
    \node[xshift=\xshift, yshift=\yshift, scale=0.5] at (\p) {\name};
  }

\draw[paperGreen, line width=0.9pt, cap=round]
  ($(\va,0.5)-(0,1.5pt)$) -- ($(\va,0.5)+(0,1.5pt)$);
  \draw[paperGreen, line width=0.5pt, cap=round]
  ($(\vb-0.007,0.5)-(0,1.5pt)$) -- ($(\vb-0.007,0.5)+(0,1.5pt)$);
\draw[paperSky, line width=0.5pt, cap=round]
  ($(\vb+0.007,0.5)-(0,1.5pt)$) -- ($(\vb++0.007,0.5)+(0,1.5pt)$);
  \draw[paperSky, line width=0.9pt, cap=round]
  ($(\vc,0.5)-(0,1.5pt)$) -- ($(\vc,0.5)+(0,1.5pt)$);

\end{tikzpicture}
    \caption{The aggregate recourse set at stage $t$ is the permutahedron $\tilde{\mathcal P}_t=\Pi(\tilde v_t)$. 
For illustration we show its projection onto the $(\tilde u_t,\tilde u_{t+1})$-plane (the full set lives in $\mathbb{R}^3$ when $n_t=3$). 
The admissible first-stage action is $\tilde u_t\in[\tilde v_t^{\,1},\tilde v_t^{\,3}]$; the breakpoints 
$\tilde v_t^{\,1}<\tilde v_t^{\,2}<\tilde v_t^{\,3}$ induce the intervals 
$[\tilde v_t^{\,1},\tilde v_t^{\,2}]$ and $[\tilde v_t^{\,2},\tilde v_t^{\,3}]$ on which the recourse function is affine. 
These intervals are precisely the projections of facets of $\Pi(\tilde v_t)$ onto the $\tilde u_t$ axis.
}
    \label{fig:perm_proj}
\end{figure}

\subsection{Aggregate Charging Model}
Let $\tilde{u}$ denote the aggregate power consumption of the population of devices in stage $t$, such that $\tilde{u}_t := \sum_{l \in \mathcal{N}} u_{t,l}$.
The set of feasible aggregate future charging profiles for the population is formalised in the following definition.
\begin{definition}[Aggregate Recourse Set]
Given the state of charge of each device $\bm{x}_t$, for stage $t\in\mathcal T$, the \emph{aggregate recourse set}, denoted $\tilde{\mathcal{P}}_t$, is the set of all feasible \emph{future aggregate} charging profiles:
\begin{equation*}
     \tilde{\mathcal{P}}_t(\bm{x}_{t-1}) :=  \left\{ \tilde{u} = \sum_{l \in \mathcal{N}} u_{t,l} \in \mathbb{R}^{n_t} \middle\vert \;\; u_l \in \mathcal{P}_{t,l}(x_{t-1, l})
    \right\}.
\end{equation*}
\end{definition}
\noindent
By definition, $\tilde{\mathcal{P}}_t$ is the Minkowski sum of the $\mathcal{P}_{t,l}(x_{t-1, l})$. Using the Minkwoski sum result of \eqref{eq:perm_mink} and the characterization of $\mathcal{P}_{t,l}(x_{t-1, l})$ as a permutahedra from \cref{thm:perm_is_flex}, and defining 
\begin{equation}\label{eq:aggregated_state_def}
    \tilde{v}_t(\bm{x}_t) := \sum_{l \in \mathcal{N}} v_{t,l}(x_{t,l}),
\end{equation}
the aggregate recourse set is the permutahedra:
\begin{equation}\label{eq:agg_is_perm}
     \tilde{\mathcal{P}}_t(\bm{x}_{t-1})= \Pi(\tilde{v}_t).
\end{equation}
\noindent
As before, the feasible instantaneous power for the population at stage-$t$ is the projection of this onto $\{t\}$:
\begin{equation}\label{eq:agg_ev_admissible}
    \mathcal{U}_t(\bm{x}_{t-1}) = \proj_{\{t\}}\big(\tilde{\mathcal{P}}_t(\bm{x}_{t-1})\big) = [\tilde{v}_{t-1}^1, \tilde{v}_{t-1}^{n_t}].
\end{equation}
This allows us to succinctly characterise the state of the population in terms of aggregated variables $\tilde{v}_t$. For \(n_t=3\),  \cref{fig:perm_proj} visualises $\tilde{\mathcal{P}}_t$, showing the admissible aggregate action of \(\tilde u_t\in[\tilde v_t^{\,1},\tilde v_t^{\,3}]\).

\subsection{Optimal Disaggregation}
Now given an aggregate charging profile, the aggregator must decide on a policy that disaggregates this into individual actions for EVs in the population. That is given $\tilde{u}_t$, we must find a policy, $\sigma: \;\; \tilde{\mathcal{U}}(x_t) \rightarrow \mathcal{U}(x) \; \times ... \times \; \mathcal{U}(x_{t,N}) $, such that 
\begin{equation*}
    \sigma(\tilde{u}_t) = (u_{t,1}, ..., u_{t,N}), \quad \text{s.t.} \quad \tilde{u}_t = \sum_{l \in \mathcal{N}} u_{t,l}.
\end{equation*}
This disaggregation policy will update the individual states, and hence induce a transition function on the aggregate state variables via \eqref{eq:aggregated_state_def}. For an aggregated state $\tilde{v}_t$, we denote the transition function induced by $\sigma$ as $T_{\sigma}(\tilde{v}_{t-1}, \cdot)$.

\begin{lemma}\label{lemma:opt_diagg_policy}
    There exists an optimal disaggregation policy $\sigma^*$, such that $\forall \; \tilde{u}_t \in \mathcal{U}_t(\bm{x}_{t-1})$ and for all other policies $\sigma$:
    \begin{equation*}
        \mathcal{V}_t(T_{\sigma^*}(\tilde{v}_{t-1}, \tilde{u}_t)) \leq \mathcal{V}_t(T_{\sigma}(\tilde{v}_{t-1}, \tilde{u}_t)),
    \end{equation*}
    The transition function induced by $\sigma^*$ is given by 
    \begin{equation*}
        T^i_{\sigma^*}(\tilde{v}_{t-1}, \tilde{u}_t) = 
            \begin{cases}
                \tilde{v}^{i+1}_{t-1}                                      & \tilde{u}_t < \tilde{v}^{i+1}_{t-1} \\
                \tilde{v}^i_{t-1} + \tilde{v}^{i+1}_{t-1}  - \tilde{u}_t   & \tilde{v}^i_{t-1} \leq \tilde{u}_t \leq \tilde{v}^{i+1}_{t-1} \\
                \tilde{v}^i_{t-1}                                          & \tilde{u}_t > \tilde{v}^i_{t-1}
            \end{cases}
    \end{equation*}
\end{lemma}

\subsection{MSLP Reformulation}
With the aggregate variables and the optimal disaggregation map $\sigma^\star$, we can abstract away the details of the device-level decisions. Any feasible aggregate power consumption, $\tilde u_t$, admits an optimal disaggregation. Thus, under a single system price, the objective depends only on $\tilde u_t$. Hence, we can rewrite the MSLP in terms of these aggregate coordinates, with admissible set: $\mathcal U_t(\bm x_{t-1})=\proj_{\{t\}}\tilde{\mathcal P}_t(\bm x_{t-1})$, and dynamics given by $T_{\sigma^\star}$. The resulting problem is an equivalent but lower-dimensional MSLP.
With the boundary condition $\mathcal{V}_{T+1}(\tilde{v}_T) = 0$, the recourse function of this can be written as 
\begin{equation}\label{eq:agg_recourse}
    \begin{aligned}
    V_t(\tilde{v}_{t-1}, c_t) = \min_{\tilde{u}_t \in \mathbb{R}} \quad &c_t \tilde{u}_t + \mathcal{V}_{t+1}(\tilde{v}_t) \\[4pt]
\text{s.t.}\quad
& \tilde{u}_t \in [\tilde{v}^1_{t-1}, \tilde{v}^{n_{t-1}}_{t-1}], \\[2pt]
& \tilde{v}_t = T_{\sigma^\star}(\tilde{v}_{t-1}, \tilde{u}_t), \\[2pt]
\end{aligned}
\end{equation}

\section{Affine Regions of the recourse function}

As noted in \cref{rem:affine_region}, the domain of the recourse function admits a finite partition over which it is affine. In this section, we derive an explicit representation for these regions. We first establish that the recourse function is affine in the aggregate state variables $\tilde v$. Leveraging the transition map $T_{\sigma^\star}$ from the previous section, we then translate this result into a partition in the aggregate decision $\tilde u$ and give closed-form expressions for the corresponding affine coefficients (slopes and intercepts) on each region.

\begin{lemma}\label{lemma:affine_recourse}
    The recourse function is affine in $\tilde{v}_t$
    \begin{equation}\label{eq:affine_recourse}
        \mathcal{V}_{t+1}(\tilde{v}_{t}) = \sum_{i=1}^{n_{t}} w_{t}^i \tilde{v}_t^i
    \end{equation}
    where for $w_{t} \in \mathbb{R}^{n_t}$, and for all $i = 1,..., n_t$:

\begin{equation*}
    \begin{aligned}
        w_{t}^i  &= \mathbb{P}(c_{t+1} \leq w_{t+1}^i) w_{t+1}^i
                    + \mathbb{P}(w_{t+1}^{i-1} \leq c_{t+1}) w_{t+1}^{i-1}\\
             & \hspace{3pt} +\mathbb{P}( c_{t+1} \in [w_{t+1}^i, w_{t+1}^{i-1}])\,\mathbb{E}\left[\,c_{t+1} \mid  c_{t+1} \in [w_{t+1}^i, w_{t+1}^{i-1}]\right]
    \end{aligned}
\end{equation*}
    with terminal condition $w^1_T =  \mathbb{E}[c_T]$ where $w_T \in \mathbb{R}$.
\end{lemma}
\begin{remark}
    Note how the structure of \eqref{eq:affine_recourse} is identical to the optimal value of solving an LP over a permutahedron, as reviewed in \eqref{eq:perm_lp}. By definition, the coefficients $w^i_t \geq w^{i+1}_t$ are nonincreasing. Intuitively, they may be understood as the ordered expected future cost, under the non-anticipativity constraints of the problem. They represent, the expected marginal cost of allocating one more unit of aggregate flexibility to $\tilde v_t^i$.
\end{remark}
\noindent
Using the transition map from \cref{lemma:opt_diagg_policy}, the recourse as a function of the aggregate decision is
\begin{equation}\label{eq:recourse_transition}
\mathcal{V}_{t+1}(\tilde u_t)
= \sum_{i=1}^{k} w_t^{i}\,\tilde v_{t-1}^{\,i}
  + \sum_{i=k}^{n_t} w_t^{i}\,\tilde v_{t-1}^{\,i+1}
  - w_t^{k}\,\tilde u_t,
\end{equation}
for all $\tilde{u}_t \in \left[\tilde v_{t-1}^k, \tilde v_{t-1}^{k+1}\right]$, for $k=1,\ldots,n_t-1$.
Thus $\mathcal{V}_{t+1}$ is piecewise affine in $\tilde u_t$ with breakpoints at $\{\tilde v_{t-1}^{\,i}\}_{i=1}^{n_t}$.
The following theorem formalises this partitioning and provides the closed-form coefficients.

\begin{theorem}\label{thm:affine_regions}
    Let $\mathcal{C}_t$ denote a partitioning of $\mathbb{R}$, such that
    \begin{equation*}
        \mathcal{C}_t =\left\{\left[\tilde v_{t-1}^k, \tilde v_{t-1}^{k+1}\right]\; \middle| \; k = 1,...,n_t-1\right\}
    \end{equation*}
    For all partitions $\mathcal{R}_k \in \mathcal{C}_t$, and for all $\tilde u_t \in \mathcal{R}_k $, the recourse function is affine and is given by
    \begin{equation*}
        \mathcal{V}_{t+1}(\tilde{u}_{t}) = \alpha_{t, k} + \beta_t^{k} \tilde u_t
    \end{equation*}
    \begin{equation*}
       \text{where} \quad \alpha_{t, k} = \sum_{i=1}^k w^i_t \tilde{v}^i_{t-1} + \sum_{i=k}^{n_t} w^i_t \tilde{v}^{i+1}_{t-1} \quad \text{and} \quad \beta_t^k = - w^k_t.
    \end{equation*}
\end{theorem}
\begin{proof}
    This follows directly from \cref{lemma:affine_recourse}, and the form of the recourse function from \eqref{eq:recourse_transition}.
\end{proof}

\begin{remark}
One can interpret the partitioning $\mathcal{C}_t$ as the chamber complex generated by projecting the faces of $\tilde{\mathcal{P}}_t$ onto $\tilde{u}_t$ \cite{Forcier2024ExactProblems}. Using its characterisation as a permutahedra from \eqref{eq:agg_is_perm}, the partitions are exactly the projection of the faces of $\tilde{\mathcal{P}}_t$ onto $\tilde{u}_t$ \cite{Frank2011ConnectionsOptimization}, this is depicted in \cref{fig:perm_proj}.
\end{remark}

\section{Numerical Results}
Here we numerically assess the proposed work. We first validate accuracy by comparing our exact recourse with SDDP approximations on a toy instance, then quantify runtime/convergence across horizons, and finally demonstrate operational use via large-scale price–response curves built from GB imbalance-price data.

\textit{1) Exact vs. SDDP Recourse Function Comparison.} 
We begin with a toy model of, $N=10$ EVs, and $T=5$ stages, to validate the theory. We synthesise a stage-dependent price distribution for $\{c_t\}$ and compute the \emph{exact} expected recourse by assembling the partitions and coefficients from \cref{thm:affine_regions}. Using the same initial state, we then run SDDP for 2 and 10 iterations and plot the resulting approximations on the recourse function.
\Cref{fig:value_function_comparison} shows the exact curve alongside the SDDP approximations. As expected, the SDDP recourse functions under-estimate the true recourse function, and tighten monotonically with iterations; by 10 iterations the approximation is nearly indistinguishable. In this toy setting several regions share identical slopes, so different cuts coincide; in larger instances such coincidences are rarer.

\begin{figure}[t]
    \centering
    \includegraphics[width=0.48\textwidth]{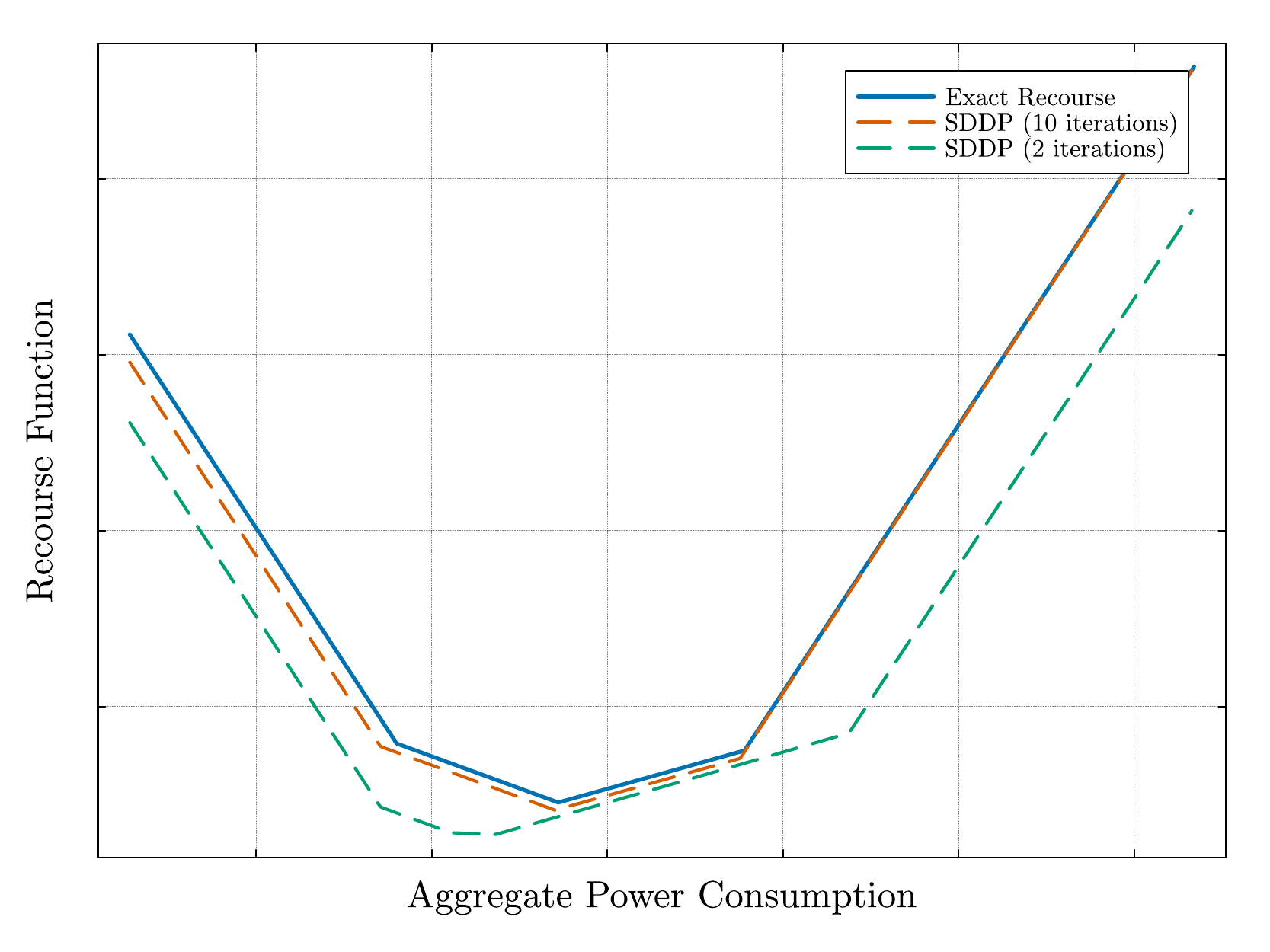}
    \caption{Exact expected recourse (solid blue) versus SDDP approximations for a toy instance with $N=10$ EVs and $T=5$ stages. The exact curve is computed from the analytic partition and coefficients of \cref{thm:affine_regions}; dashed curves show SDDP minorants after 2 and 10 iterations.}

    \label{fig:value_function_comparison}
\end{figure}
\textit{2) Runtime Benchmarking.} 
We benchmark our approach against SDDP on instances with $N=10$ EVs and horizons $T\in\{12,24,36,48\}$. After each SDDP iteration we record the wall time and the \emph{expected cost} obtained by simulating the policy induced by the current cut approximation. Costs are normalised by the exact value, and the resulting runtime curves are shown in \cref{fig:runtime_curve}. 
Our method attains the exact recourse after a single backward sweep, so it reaches the optimum essentially immediately. The gap in \cref{fig:runtime_curve} reflects this: across all $T$, SDDP improves monotonically but requires orders of magnitude more time to approach the normalised optimum. For visual clarity we plot the exact runtime curve only for $T=48$, the longest horizon. In these experiments we discretise each $c_t$ with 100 samples. While our construction can handle continuous distributions, and is largely insensitive to the number of samples. Since each backward pass solves subproblems for sampled realizations, SDDP’s runtime grows with both the horizon and the number of samples per stage.

\textit{3) Case Study: Price–Response Curves at Scale.} 
To illustrate the practical use of this work, we compute the optimal \emph{aggregate} consumption as a function of the observed system price. We consider a fleet of $50{,}000$ EVs and a horizon of $48$ settlement periods. Priors for future prices are built from historic GB system-price data (gathered from Elexon \cite{ElexonPortal}) by constructing stage-dependent empirical distributions for each settlement period.
For each observed first stage price $c_1$, we solve the first-stage problem in \eqref{eq:agg_recourse} with aggregate state $\tilde v_1$ to obtain $\tilde u_1(c_1)$. Sweeping over a range of $c_1$ yields the price–quantity curves in \cref{fig:cost_curve}.
To illustrate the diversity in potential cost curves, arising from different priors over future system price we construct two future  system price forecasts, one for system price in winter and summer months.
The winter curve lies above the summer curve across most prices, reflecting higher expected future prices in winter, and so it is optimal to charge more in winter for the same observed price. Note, there exist far more sophisticated probabilistic forecasts for the system price, which can be plugged directly into the work proposed here. The curves saturate at the fleet’s aggregate power bounds and can be used directly as dispatch policies for the aggregator or as bid/offer curves in reserve markets (e.g., mFRR).

\begin{figure}[t]
    \centering
    \includegraphics[width=0.48\textwidth]{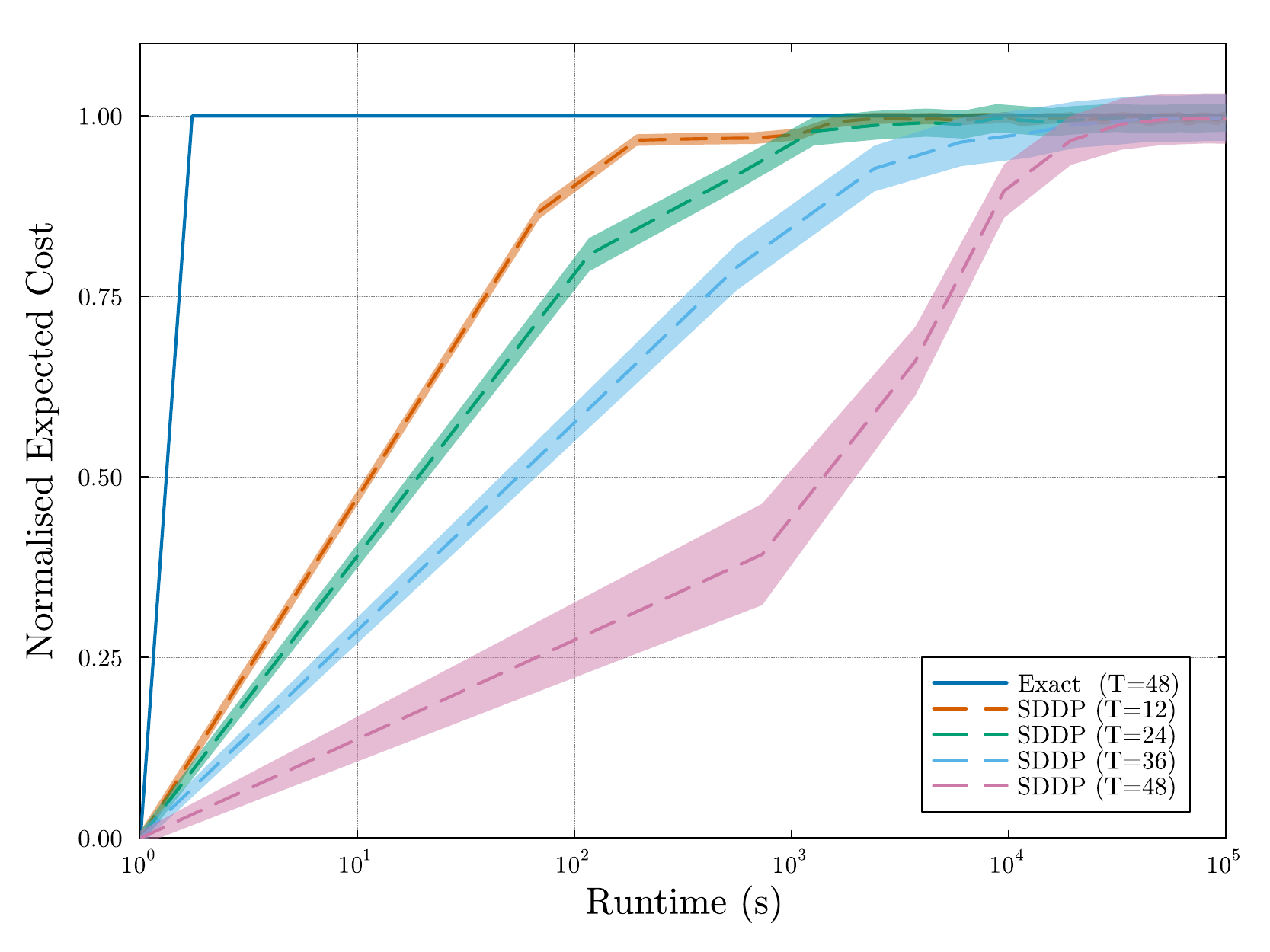}
\caption{Normalised expected cost versus wall-clock time (log scale) for $N=10$ EVs. The blue exact curve reaches the optimum after a single backward sweep, while SDDP curves improve monotonically but converge slower.}
    \label{fig:runtime_curve}
\end{figure}

\section{Conclusion}
We presented a geometric reformulation of the EV MSLP imbalance optimisation problem that replaces a high-dimensional device model with an exact aggregate representation. Exploiting the fact that the aggregate recourse set is a permutahedron, we derived an optimal disaggregation policy, an explicit state transition in aggregate coordinates, and derived an explicit form for a partitioning over which the expected recourse function is piecewise affine. The slopes and intercepts of all regions admit closed-form expressions via truncated expectations of future prices, yielding an exact recourse function and first-stage policy without iterative cut generation

This lays the groundwork for extensions in which population level constraints couple devices; for example, enforcing demand-response commitments specified by the TSO or to manage feeder constraints. Future directions include generalising to more heterogeneous device populations, incorporating risk-averse objectives (e.g., CVaR), or modelling correlated and non-stationary price processes. Overall, the paper shows that aggregating device flexibility as permutahedra, and leveraging the properties of these polytopes, removes the computational bottleneck that hampers standard SDDP on large EV populations, delivering exact policies with clear structure and practical runtimes.



\begin{figure}[t]
    \centering
    \includegraphics[width=0.48\textwidth]{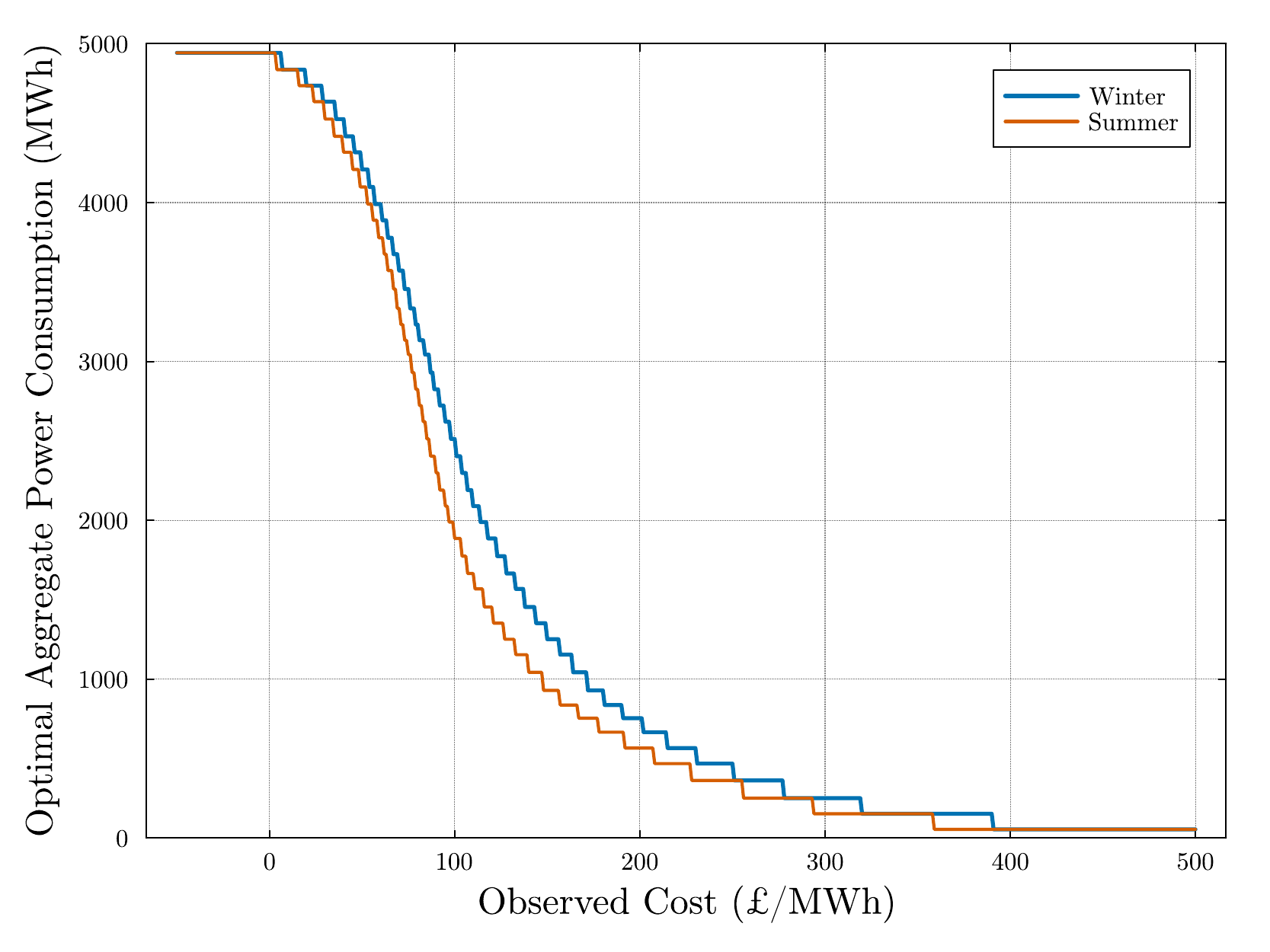}
    \caption{Price–response curves for a fleet of $50{,}000$ EVs over $48$ settlement periods: optimal aggregate consumption $\tilde u_1(c_1)$ versus the initial observed system price $c_1$.}
    \label{fig:cost_curve}
\end{figure}

\bibliographystyle{IEEEtran}
\bibliography{references}
\appendix
\subsection{Proof of \cref{lemma:opt_diagg_policy}}

\begin{proof}
There is a canonical ordering that prioritises which EVs to charge first \cite[Lemma III.5]{Evans2018RobustlyResources}, this induces the transition function $T_{\sigma^*}(\tilde{v}_{t-1}, \tilde{u}_t)$. Let $\tilde{v}^{\sigma}_{t}= T_{\sigma}(\tilde{v}_{t-1}, \tilde{u}_t)$ and 
$\tilde{v}^{\sigma^*}_{t} = T_{\sigma^*}(\tilde{v}_{t-1}, \tilde{u}_t)$. For all other policies, $\sigma$, $\tilde{v}^{\sigma}_{t} \preceq \tilde{v}^{\sigma^*}_{t}$. By \eqref{eq:majorization}, this implies $\Pi(v_t^\sigma) \subseteq \Pi(v_t^{\sigma^*})$. Therefore, the aggregate recourse set induced by $\sigma^*$ contains the aggregate recourse set of induced by any other $\sigma$. Because the objective is linear in aggregate power, enlarging the feasible set cannot increase the minimum cost; hence $\mathcal{V}_t(\tilde{v}^{\sigma^*}_{t}) \leq \mathcal{V}_t(\tilde{v}^{\sigma}_{t})$.
\end{proof}

\subsection*{Proof of \cref{lemma:affine_recourse}}
\begin{proof}
We prove this by backward induction in \(t\).
\noindent
\emph{Base case.}
When \(t=T\) we have \(n_T=1\) and
\(\mathcal V_{T}(\tilde v_{T-1})=\mathbb E[c_T]\tilde v_{T-1}^1\).
Hence \(w_T^1=\mathbb E[c_T]\), as stated.
\noindent
\emph{Inductive step.}
Fix \(t\le T-1\) and assume
\(\mathcal V_{t+2}(\tilde v_{t+1})=\sum_{i=1}^{n_{t+1}} w_{t+1}^i \tilde v_{t+1}^i\)
with \((w_{t+1}^i)_{i}\) nonincreasing.
Let \(\tilde v_t\) be given. By \cref{lemma:opt_diagg_policy},
for any \(\tilde u_{t+1}\in[\tilde v_t^j,\tilde v_t^{j+1}]\) we have
\[
\mathcal V_{t+2}(T_{\sigma^\star}(\tilde v_t,\tilde u_{t+1}))
= A_j(\tilde v_t)- w_{t+1}^j\,\tilde u_{t+1},
\]
where 
\[
A_j(\tilde v_t):=\sum_{i=1}^{j} w_{t+1}^i \tilde v_t^i
+ \sum_{i=j+1}^{n_t} w_{t+1}^{i-1}\tilde v_t^i .
\]
Thus, using \eqref{eq:agg_recourse},
\[
V_{t+1}(\tilde v_t,c_{t+1})
= \min_{j}\;\min_{\tilde u\in[\tilde v_t^j,\tilde v_t^{j+1}]}
\bigl\{A_j(\tilde v_t) + (c_{t+1}-w_{t+1}^{j})\,\tilde u\bigr\}.
\]
Since the inner objective is affine in \(\tilde u\) with slope \(c_{t+1}-w_{t+1}^{j}\),
the minimiser over the interval is its left endpoint if the slope is nonnegative and
the right endpoint otherwise. Because \((w_{t+1}^i)_i\) is nonincreasing, there is a
unique index
\[
j^\star=j^\star(c_{t+1})\in\{1,\ldots,n_t\}
\;\text{such that}\;
c_{t+1}\in [w_{t+1}^{j^\star},\,w_{t+1}^{j^\star-1})
\]
(with the conventions \(w_{t+1}^{0}=+\infty\), \(w_{t+1}^{n_t}=-\infty\)).
Therefore the optimizer is \(\tilde u^\star=\tilde v_t^{\,j^\star}\) and
\[
V_{t+1}(\tilde v_t,c_{t+1})
= \sum_{i=1}^{j^\star-1} w_{t+1}^i \tilde v_t^i
+ \sum_{i=j^\star+1}^{n_t} w_{t+1}^{i-1}\tilde v_t^i
+ c_{t+1}\,\tilde v_t^{\,j^\star}.
\]

Now let \(E_j:=\{\,w_{t+1}^{j}\le c_{t+1}\le w_{t+1}^{j-1}\,\}\),
\(p_{t+1}^j:=\mathbb P(E_j)\), and
\(\bar c_{t+1}^{\,j}:=\mathbb E[c_{t+1}\mid E_j]\).
Taking conditional expectation and then averaging over \(j\),
\[
\mathcal V_{t+1}(\tilde v_t)
=\sum_{j=1}^{n_t} p_{t+1}^j\;
\mathbb E\!\left[V_{t+1}(\tilde v_t,c_{t+1})\mid E_j\right]
= \sum_{i=1}^{n_t} w_t^i\,\tilde v_t^i,
\]
with 
\begin{equation*}
    \begin{aligned}
        w_{t}^i  &= \mathbb{P}(c_{t+1} \leq w_{t+1}^i) w_{t+1}^i
                    + \mathbb{P}(w_{t+1}^{i-1} \leq c_{t+1}) w_{t+1}^{i-1}\\
             & \hspace{3pt} +\mathbb{P}( c_{t+1} \in [w_{t+1}^i, w_{t+1}^{i-1}])\,\mathbb{E}\left[\,c_{t+1} \mid  c_{t+1} \in [w_{t+1}^i, w_{t+1}^{i-1}]\right]
    \end{aligned}
\end{equation*}
for $i=1,\ldots,n_{t-1}$. By inspection $w_{t}^i \geq w_{t}^{i+1} \;\; \forall i$, as required.
\end{proof}

\end{document}